\documentclass[conference]{IEEEtran}

\usepackage{fullpage}

\usepackage{graphicx}
\usepackage{wrapfig}
\usepackage{epsfig}
\usepackage{graphicx}
\usepackage{color}
\usepackage{mdframed}
\usepackage{mathrsfs} 
\usepackage{amsfonts}
\usepackage{latexsym}
\usepackage{amsmath}
\usepackage{amssymb}
\usepackage{color}
\usepackage{float} 
\usepackage{floatflt} 
\usepackage{caption}
\usepackage{subcaption}

\usepackage{pdfpages}

\newcommand{\suppress}[1]{}

\newtheorem{theorem}{Theorem}[section]

\newtheorem{claim}{Claim}[section]

\newtheorem{example}{Example}[section]
\newtheorem{definition}{Definition}[section]

\newtheorem{remark}{Remark}[section]

\def\e{\varepsilon}

\def\cA{\mbox{$\cal{A}$}}

\def\cW{\mbox{$\cal{W}$}}
\def\cX{\mbox{$\cal{X}$}}
\def\cY{\mbox{$\cal{Y}$}}
\def\cE{\mbox{$\cal{E}$}}

\def\cSi{{\bf \sigma}}
\def\cRa{\mbox{$\cal{R}$}}
\def\cTone{\cR_{\underline{1}}}

\def\cR{\mbox{$\cal{R}$}}
\def\cB{\mbox{$\cal{B}$}}
\def\cC{\mbox{$\cal{C}$}}

\def\cM{\mbox{$\cal{M}$}}
\def\cD{\mbox{$\cal{D}$}}

\def\cN{\mbox{$\cal{N}$}}

\def\cT{\mbox{$\cal{T}$}}

\def\e{\varepsilon}

\def\bx{{\bf x}}

\def\by{{\bf y}}

\newcommand{\ulR}{{{\underline{R}}}}

\newcommand{\ulM}{{{\underline{M}}}}
\def\uls{{\underline{\sigma}}}
\def\ulv{{\underline{v}}}
\def\ula{{\underline{a}}}
\def\ulx{{\underline{x}}}
\def\uly{{\underline{y}}}
\def\realnumbers{\mathbb{R}}

\def\HH{H}
\def\SS{S}
\def\hh{h}
\def\ss{s}
\def\JJ{\Lambda}

\def\01{\{0,1\}}

\newcommand{\remove}[1]{}

\begin{document}
\IEEEoverridecommandlockouts
\title{Beyond Capacity: The Joint Time-Rate Region}

\author{Michael Langberg\ \ \ \ \ \ \ \ \  \ \ \ \ \ \ \ \ \ 
Michelle Effros
\thanks{M. Langberg is with the Department of Electrical Engineering at the University at Buffalo (State University of New York).  
Email: \texttt{mikel@buffalo.edu}}
\thanks{M. Effros is with the Department of Electrical Engineering at the California Institute of Technology.
Email: \texttt{effros@caltech.edu}}
\thanks{This work is supported in part by NSF grants CCF-1817241 and CCF-1909451. 
}
}


\maketitle

\begin{abstract}
The traditional notion of capacity studied in the context of memoryless network communication builds on the concept of block-codes and requires that, for sufficiently large blocklength $n$, all receiver nodes simultaneously decode their required information  after $n$ channel uses. 
In this work, we generalize the traditional capacity region by exploring communication rates achievable when some receivers are required to decode their information before others, at different predetermined times; referred here as the {\em time-rate} region.
Through a reduction to the standard notion of capacity, we present an inner-bound on the time-rate region.
The time-rate region has been previously studied and characterized for the memoryless broadcast channel (with a sole common message) under the name {\em static broadcasting}.  
\end{abstract}

\section{Introduction} 
\label{sec:intro}
In the context of communication over multi-source multi-terminal memoryless channels (i.e., networks), one traditionally seeks the design of communication schemes that, for a given blocklegth $n$, allow the successful decoding of source information at receiver nodes after $n$ channel uses. 
Roughly speaking,\footnote{All concepts mentioned in this section are defined in detail in Section~\ref{sec:model}.} rate vector $\ulR = (R_1,\dots,R_k)$, is said to be achievable with blocklength $n$ and decoding error $\e>0$ over a given $k$-source network, if for uniformly distributed $R_i n$-bit messages $m_i$ (for $i=1,\dots,k$) there exists a communication scheme that after $n$ channel uses allows all receivers to decode their required source information with success probability at least $1-\e$.
The capacity region of the communication problem at hand describes the closure of all rate vectors $\ulR$ achievable with asymptotic blocklength and vanishing error
(see, e.g., \cite{el2011network}).

In this work, we generalize the notion of capacity and study communication in the setting in which some network nodes are required to decode their information before others.  
More precisely, we study communication schemes of blocklegnth $n$ in which receiver $v_j$ is required to decode the $R_i n$-bit message $m_i$ after $\sigma_{ij}n$ channel uses, where $\sigma_{ij}$ is a predetermined {\em time constraint} less than or equal to 1.
The traditional capacity region is captured when all constraints $\sigma_{ij}$ equal 1.
Different values for  parameters $\sigma_{ij}$ represent settings in which certain receivers are required to decode earlier than others due to, e.g., time-sensitive information, physical receiver constraints such as battery life, or computational receiver constraints that are due to parallel communication or processing tasks.
For example, in the setting of IoT, a base station with side information including the remaining battery life of the sensors under its control may design broadcast codes that allow earlier decoding at low-battery sensors; in disaster areas, a base station with event information may design broadcast codes that allow earlier decoding at receivers close to a critical event.

To represent the achievable rates $\ulR=(R_1,\dots,R_k)$ under 0/1 demand matrix $S$ and time constraints $\uls = (\sigma_{ij} : \ss_{ij}=1)$, in this work we define and study the joint {\em time-rate} region $\cT$ which is the closure of vector pairs $(\uls,\ulR)$ for which rate $\ulR$ is achievable with time constraints $\uls$. Here, for the communication problem at hand,  demand matrix $S=[s_{ij}]$ sets $s_{ij}=1$ if receiver $v_j$ wants message $m_i$ and 0 otherwise.
Rate $\ulR$ is achievable with time constraints $\uls$ on network $\cN$, if there exists a blocklength-$n$ communication scheme for $\cN$ such that for $\ss_{ij}=1$, receiver $v_j$ can decode the $R_i n$-bit message $m_i$ (with high probability) after $\sigma_{ij}n$ channel uses (see Section~\ref{sec:model}, and in particular Definition~\ref{def:timerate}, for formal details).

It is convenient to represent the time-rate region by considering its {\em slices} with respect to $\uls$ or alternatively with respect to $\ulR$. For the former, consider expressing $\cT$ by the collection $\{\cRa_\uls\}_{\uls}$.
Here, for any setting of time constraints $\uls$, $\cRa_\uls=\{\ulR \mid (\uls,\ulR) \in \cT\}$ is the closure of the set of rate vectors $\ulR$ that are achievable with blocklength $n$. 
Achievability implies that for $(i,j)$ with $\ss_{ij}=1$ receiver $v_j$ can decode message $m_i$ (with high probability) after $\sigma_{ij}n$ channel uses. Using this notation, the standard capacity region, in which all time constraints $\sigma_{ij}$ equal 1, is denoted by $\cTone$.
Given $\uls$, the region $\cRa_\uls$ captures the {\em tradeoff in message rates} $\ulR$ achievable given a collection of   decoding time-constraints implied by $\uls$. 

For the latter, one may express $\cT$ by the collection $\{\cSi_\ulR\}_{\ulR}$, where for any rate vector $\ulR=(R_1,\dots,R_k)$, $\cSi_\ulR=\{\uls \mid (\uls,\ulR) \in \cT\}$ is the closure of the set of time-constraints $\uls$ that allow the communication of rate-$\ulR$ messages.
Given a fixed rate vector $\ulR$, the region $\cSi_\ulR$ represents the {\em tradeoff in decoding times} $\uls$ supporting the communication of messages of predetermined rate represented by $\ulR$. 

The characterizations $\cRa_\uls$ and $\cSi_\ulR$ are equivalent in the sense that each suffices to recover the time-rate region $\cT$.
While the perspectives represented by $\cRa_\uls$ and $\cSi_\ulR$ both have operational significance, in this work, we focus mainly on the study of $\cRa_\uls$.
Like the standard capacity region $\cTone$, for any $\uls$, the region $\cRa_\uls$ is convex by the usual time-sharing argument (here, codes should be interleaved). 
As a result, $\cRa_\uls$ lends itself more naturally to our analysis.
This is in contrast to the region $\cSi_\ulR$, which is not necessarily convex (since time sharing between a code that delivers rate $\ulR$ at time constraints $\uls$ and a code that delivers rate $\ulR$ at time constraints $\uls'$ does not always create a code that delivers rate $\ulR$ at time constraints $\alpha\uls+(1-\alpha)\uls'$).
The latter is shown, e.g.,  in \cite{shulman2000static, Shulman03}, for the two-terminal broadcast channel.

In this work, we study the joint time-rate region for general multiple-source multiple-terminal memoryless networks $\cN$. 
This manuscript is structured as follows.
In Section~\ref{sec:model}, we define our model in detail.  
In Section~\ref{sec:prior}, we outline prior related works, focusing on {\em static broadcasting} (initiated in \cite{shulman2000static, Shulman03}) which studies the time-rate region in the single-source broadcast setting (i.e., $k=1$), and rateless codes, e.g., \cite{luby2002lt,mackay2005fountain,shokrollahi2006raptor}.
Our main result appears in Section~\ref{sec:inner}, where, given any network $\cN$ and set of time constraints $\uls$, we design an inner bound on $\cRa_\uls(\cN)$ through a reduction to the study of {\em traditional} capacity regions on related networks.
To put our results in perspective, in Section~\ref{sec:results_example}, we present a single-message network for which our inner bound is not tight.
Finally, we conclude in Section~\ref{sec:conclude}.

\section{Model} 
\label{sec:model}


%

We use the following notational conventions. Scalars are represented by lowercase letters, e.g., $v$; vectors by underlined lowercase letters, e.g., $\ulv$; matrices and sets are represented by uppercase letters; script letters typically denote alphabets, e.g., $\cX$, or more complex structures; and random variables are denoted in bold.
For a positive real value $r$, $[r]$ represents the set $\{1,2,\dots,\lfloor r \rfloor\}$.

\subsection{Memoryless Communication Channel}

Our model for a given discrete memoryless communication channel $\cW$ comprises:
\begin{itemize}
\item {\bf Nodes:} a collection of $\ell$ network communication nodes $V=(v_j : j \in [\ell])$.
\item {\bf Alphabets:} each $v_j \in V$ observes channel outputs from alphabet $\cY_j$ and transmits channel inputs in alphabet $\cX_j$. 
\item {\bf Channel:} Let $W(\uly|\ulx)$ be a conditional probability distribution of $\uly = (y_j : j \in [\ell]) \in \prod_{j \in [\ell]} \cY_j$ given $\ulx = (x_j : j \in [\ell]) \in \prod_{j \in [\ell]} \cX_j$. 
\end{itemize}
Thus, the channel $\cW$ is represented by the tuple 
$$\cW=(\prod_{j \in [\ell]} \cX_j,W(\uly|\ulx),\prod_{j \in [\ell]} \cY_j).$$

\subsection{The Message Set, Message Side-information, and Requirements}
We denote the collection of source messages to be communicated over channel $\cW$ by $\ulM=(m_i : i \in [k])$.
The message side-information is defined by the $k \times \ell$ binary matrix  $\HH=[\hh_{ij}]$ for which $\hh_{ij}=1$ if and only if message $m_i$ is available to node $v_j$ at the start of the communication process. Similarly demand matrix $\SS=[\ss_{ij}]$ is a $k \times \ell$ binary matrix for which $\ss_{ij}=1$ if and only if node $v_j$ requires message $m_i$.

\subsection{Network Communication Problem}
Combining the elements above, a network communication problem $\cN$ is defined by the tuple $(\cW,\ulM,\HH,\SS)$.

\subsection{Network code}
Let $\cN=(\cW,\ulM,\HH,\SS)$ be a network communication problem as above.
Let $n$ be an integer and $\ulR=(R_1,\dots,R_k)$ be a rate vector.
For $i \in [k]$, let $\cM_i=[2^{R_i n}]$ be the message alphabet of $m_i \in \ulM$.
A $(\ulR,n)$ code $\cC$ for communication problem $\cN$ consists of the following components:
\begin{itemize}
\item {\bf Encoders:} with each node $v_j \in V$ we associate a time-varying encoder, which at time $\tau$ is defined as
$$
E_j^\tau: \prod_{\hh_{ij}=1}\cM_i  \times \cY_j^{[\tau-1]} \rightarrow \cX_j.
$$
\item {\bf Decoders:} with each node $v_{j} \in V$ we associate a time-varying decoder, which at time $\tau$ is defined as
$$
D_j^\tau: \prod_{\hh_{ij}=1}\cM_j  \times \cY_j^{[\tau]} \rightarrow \prod_{\ss_{ij}=1}\cM_i.
$$
\end{itemize}
Thus a code $\cC$ for $\cN$ is defined by the tuple 
$$(\cE,\cD) \triangleq ((E_j^\tau : j \in [\ell], \tau \in \mathbb{N}),(D_j^\tau: j \in [\ell], \tau \in \mathbb{N})).$$

\noindent
{\bf Achievability:} Let $\uls = (\sigma_{ij}: \ss_{ij} = 1)$  for $\sigma_{ij} >0$ and let $\e>0$.
A $(\ulR,n)$ code $\cC$ for network communication problem $\cN$ is said to be an $(\e,\uls,\ulR,n)$ code if it allows successful decoding with probability at least $1-\e$.
Specifically, given independent messages ${\bf m}_i$, $i \in [k]$, uniformly distributed over $\cM_i = [2^{R_i n}]$, operating code $\cC$ over channel $\cW$ yields a time-$\tau$ channel output ${\bf{\uly}}^\tau$ for which 
{\small{
$$
\Pr[\forall (i,j) \ \text{s.t.} \ s_{ij}=1, D_j^{\sigma_{ij}n}(({\bf m}_i : \hh_{ij}=1), {\bf y}_j^{[\sigma_{ij} n]})={\bf m}_i),
$$}}
is at least $1-\e$. Here, the probability is taken over the randomness of the messages and the channel $W$.

\begin{remark}
In our definition above, we consider time parameters $\uls=(\sigma_{ij}: s_{ij}=1)$ for any positive values of $\sigma_{ij}$.
Although this seemingly generalizes the discussion in Section~\ref{sec:intro} in which the collection of time-parameters $\sigma_{ij}$ has a maximum value of 1 (i.e., the setting in which some receiver nodes decode at time $n$, and others earlier), it is not hard to verify that the two definitions are equivalent, and we use the former to simplify our analysis.
More precisely, our definitions imply the following tradeoff between the blocklength $n$ and the pair $(\uls,\ulR)$.
\begin{claim}
\label{claim:trade}
 Let $\cN$ be a network communication problem. 
 If $\cC$ is an $(\e,\uls,\ulR,n)$ code for $\cN$, then for any $\alpha >0$, $\cC$ is also an $(\e,\alpha\uls,\alpha\ulR,\frac{n}{\alpha})$  code for $\cN$.
\end{claim}
\end{remark}

\subsection{Time-rate region}
We now define the time-rate region $\cT(\cN)$ of network communication problem $\cN$.
\begin{definition}
[Joint time-rate region]
\label{def:timerate}
The joint time-rate region $\cT(\cN)$ of communication problem $\cN$ is the collection of $(\uls,\ulR)$  such that for every $\e>0$ and $\delta>0$, for all $n$ sufficiently large there exists an $(\e,\uls+\delta,\ulR- \delta,n)$ code for $\cN$.
Here, for a vector ${\underline{v}}$ and a scalar $\delta$, ${\underline{v}}-\delta$ represents the vector with entires $v_i-\delta$.
Equivalently, one can express the joint time-rate region $\cT$ by the collection $\{\cSi_\ulR\}_{\ulR \in \realnumbers^{k}}$ where for each $\ulR \in \realnumbers^{k}$,
$$
\cSi_\ulR(\cN)=\{\uls \mid (\uls,\ulR) \in \cT(\cN)\},
$$
or the collection$\{\cRa_\uls\}_{\uls}$ where for each $\uls \in \realnumbers_{>0}^{\|\SS\|_0}$
$$
\cRa_\uls(\cN)=\{\ulR \mid (\uls,\ulR) \in \cT(\cN)\}.
$$
Here $\|\SS\|_0$ represents the number of entries in $\SS$ that equal 1, and $\realnumbers_{>0}$ represents the positive reals. 
Using the latter notation, the traditional network capacity region, e.g., \cite{el2011network}, of $\cN$ equals $\cTone(\cN)$
for ${\underline{1}} = (1,1,1, \dots, 1)$. 
\end{definition}

\section{Related work}
\label{sec:prior}

\subsection{Static Broadcasting}
\label{sec:static}
Previous work on the time-rate region (defined in a different but equivalent manner) under the name {\em static broadcasting} treats the broadcast channel with a common message \cite{shulman2000static, Shulman03} and single-source multicast network coding \cite{lun2008coding}. In \cite{lun2008coding}, the optimal decoding time $\sigma_j$ at terminal $v_j$ is characterized by the corresponding min-cut to the source. 
In \cite{shulman2000static, Shulman03}, a single transmitter transmits a single message $m$ to a pair\footnote{For simplicity of presentation, we consider only two terminal nodes. In \cite{shulman2000static, Shulman03}, the broadcast setting with multiple terminal nodes (not necessarily two) is studied. Our discussion generalizes to multiple terminals as well.} of terminal nodes, here denoted by $t_1$ and $t_2$, over a broadcast channel $(\cX,W(y_1,y_2|x),\cY_1 \times \cY_2)$.
%
Time parameter $\uls = (\sigma_1,\sigma_2)$ represents the decoding times for message $m$ at terminals $t_1$ and $t_2$, respectively.


The time-rate region $\cT(\cN)$ is characterized in \cite{shulman2000static, Shulman03} by the collection of all $(\sigma_{1},\sigma_{2},R)$ for which there exists an $n$ and a collection of distributions $(\bx^\tau \sim p^\tau_x: \tau \in \mathbb{N})$ over alphabet $\cX$  such that 
\begin{align*}
Rn &\leq \sum_{\tau=1}^{\sigma_{1} n}I(\bx^\tau;\by_1^\tau),\\
Rn &\leq \sum_{\tau=1}^{\sigma_{2} n}I(\bx^\tau;\by_2^\tau).
\end{align*}

To simplify the representation of $\cT(\cN)$ in \cite{shulman2000static, Shulman03}, let $W(y_1,y_2|x)$ equal the channel pair
$W_1(y_1|x), W_2(y_2|x).$
Let $C_1$ and $C_2$ be the point-to-point capacities of $W_1$ and $W_2$, repectively.
Using the concavity of mutual information (over the distributions $\{p^\tau_x\}_{\tau}$), the region $\cT(\cN)$ can be described by the collection of all $(\sigma_{1},\sigma_{2},R)$ for which there exists a distribution $\bx \sim p_x$ over alphabet $\cX$ such that for $\sigma_{1} \leq \sigma_{2}$,  
\begin{align*}
R &\leq \sigma_{1} I(\bx;\by_1),\\
R &\leq \sigma_{1} I(\bx;\by_2)+(\sigma_{2}-\sigma_{1})C_2
\end{align*}
and for $\sigma_{1} \geq \sigma_{2}$, 
\begin{align*}
R &\leq \sigma_{2} I(\bx;\by_2),\\
R &\leq \sigma_{2} I(\bx;\by_1)+(\sigma_{1}-\sigma_{2})C_1
\end{align*}
This {\em two-phase} representation of $\cT(\cN)$ resembles the inner-bound methodology presented in the main result of this work (Theorem~\ref{the:N_lambda}) in which we concatenate block-codes corresponding to the different phases of communication.
For the case at hand, e.g., for $\sigma_1 \leq \sigma_2$, during the first phase of time-steps up to $\sigma_1 n$, terminal $t_1$ is required to decode the message $m$ while terminal $t_2$ receives {\em partial} information on $m$.
During the second phase of the remaining time steps up to $\sigma_2 n$, terminal $t_2$ is required to obtain additional information that allows successful decoding of $m$.

\subsection{Rateless codes}
 {\em Rateless codes}, in which reliable decoding does not occur at a predetermined time (i.e., blocklength) $n$ but may vary depending on the channel realization, are 
somewhat related to our problem.
See, for example, \cite{luby2002lt,mackay2005fountain,shokrollahi2006raptor} in the context of the erasure channel, \cite{cerf1983dod, cerf2005protocol,postel1980user,park1997highly,zhao2001tapestry,fall2011tcp} in the context of adaptive routing protocols for networks, and \cite{castura2005rateless,uppal2007practical,molisch2007performance,liu2009fountain} 
in the context of discrete memoryless networks such as multiple access, relay, and broadcast channels. 

The model and measure of quality in rateless codes, however,  differ significantly from our model.
Specifically, we assume decoding times $\uls$ that are fixed for each demand (described by a message-receiver pair $(i,j)$) and characterize message rates $\ulR$ achievable  with high probability over the channel realization.

\section{An inner bound on the time-rate region}
\label{sec:inner}
Given a network problem $\cN$, the discussion that follows proposes a {\em time-expansion} $\cN_1,\dots,\cN_{\JJ}$ of $\cN$ and then uses the standard capacities of networks in the time expansion (i.e., $\cTone(\cN_\lambda)$ for $\lambda \in [\JJ]$) to derive an inner bound on the time-rate region $\cT(\cN)$.


Consider any $(\uls,\ulR)$.
We start by defining a network $\cN_0=(\cW_0,\ulM_0,\HH_0,\SS_0)$ and a corresponding pair $(\uls_0,\ulR_0)$ such that 
$$
\ulR_0 \in \cRa_{\uls_0}(\cN_0) \ \ \Leftrightarrow \ \ \ulR \in \cRa_\uls(\cN).
$$
We then expand $\cN_0$ to $\cN_1,\dots,\cN_{\JJ}$ such that one can express an inner bound on $\cRa_{\uls_0}(\cN_0)$ by the sequence of standard capacity regions $(\cTone(\cN_\lambda) : \lambda \in \JJ)$. 
Details on our reductions follow.

\vspace{2mm}

\noindent
{\bf $\bullet$ The network $\cN_0$:}
Let $\cN=(\cW,\ulM,\HH,\SS)$ and let $\uls = (\sigma_{ij} : \ss_{ij}=1)$ be a collection of time-parameters.
We design network communication problem $\cN_0=(\cW_0,\ulM_0,\HH_0,\SS_0)$ using $\uls$.
Let $\JJ$ be the number of distinct decoding times in vector $\uls$.
The major difference between $\cN$ and $\cN_0$ is in the message sets $\ulM$ and $\ulM_0$ and in the corresponding decoding times. 

For each message $m_i \in \ulM$ in $\cN$, we design a partition $(m_{(i,\ula)}: \ula \in \cA_i)$ of $m_i$ into sub-messages $m_{(i,\ula)}$ to be transmitted in $\cN_0$.
For each $i$, the vector $\ula=(a_1,\dots,a_\ell)$ falls in the set $\cA_i \subseteq [\JJ+1]^\ell$ to be defined shortly.
The message set $\ulM_0$ of $\cN_0$ consists of messages $(m_{(i,\ula)}: i \in [k], \ula \in \cA_i)$.

\begin{figure}[t!]
\centering
\includegraphics[scale=0.25]{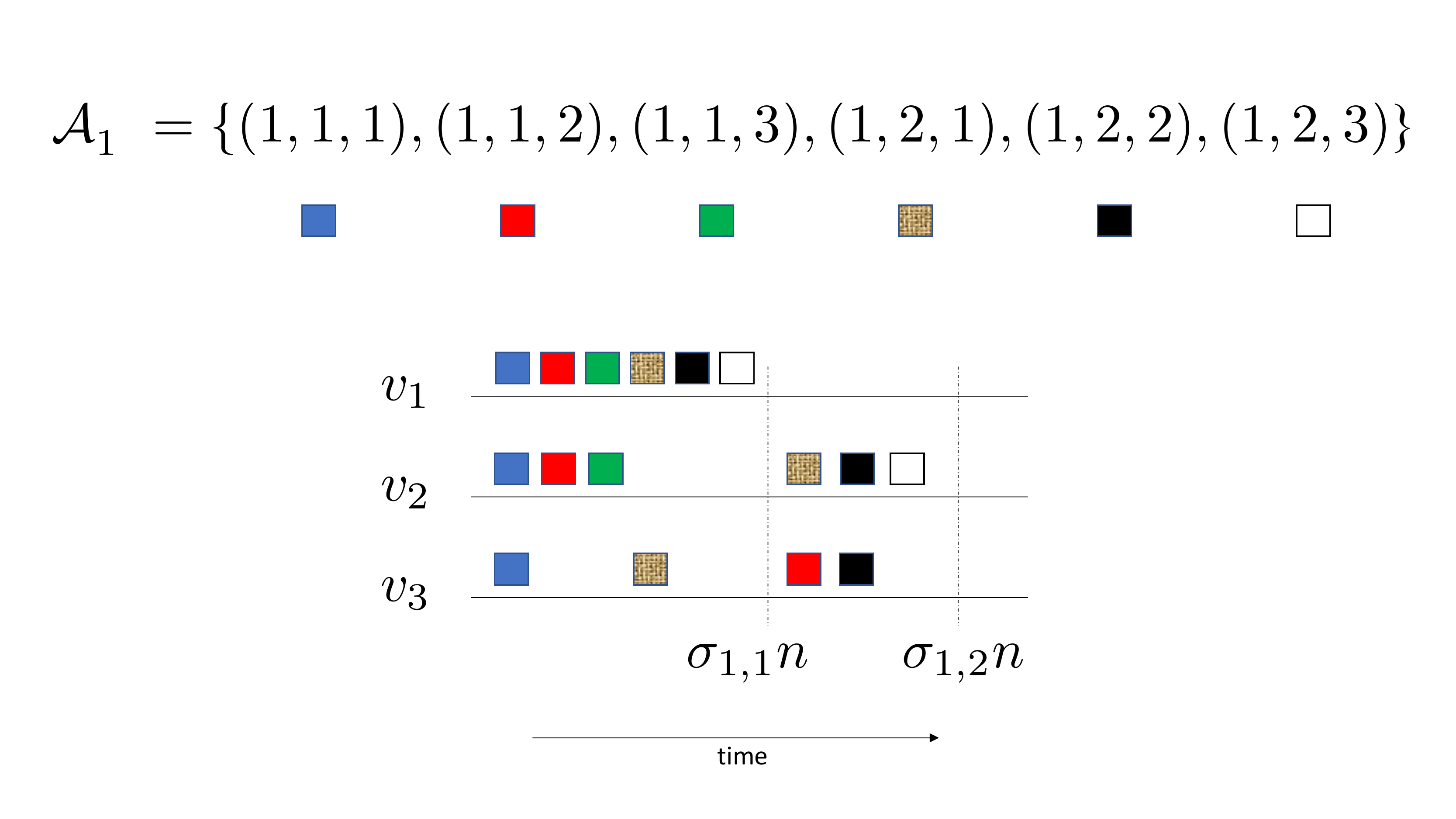}
\caption{A depiction of the message set $\ulM_0$ of $\cN_0$ from Example~\ref{ex:N0}. 
}
\label{fig:M0}
\end{figure}

Roughly speaking, in $\cN_0$ we require the sub-messages $m_{(i,\ula)}$ of message $m_i$ to be decoded at or before the decoding times for $m_i$ in $\cN$. 
Specifically, for $i \in [k]$, $\ula=(a_1,\dots,a_\ell) \in \cA_i$, and $j \in [\ell]$, the parameter $a_j$ represents the updated decoding time for $m_{(i,\ula)}$ at $v_j$.
The partition of $m_i$ and the updated decoding times,  govern our inner-bound on $\cT(\cN)$ through $\cN_0$ and $\cN_1,\dots,\cN_\JJ$. 
Details follow.

Consider sorting the distinct values of time-parameters in $\uls$ in increasing order. Let $\sigma_1< \sigma_2 < \dots < \sigma_\JJ$ be these distinct values. For $\lambda \in [\JJ]$, we define the set $\Delta(\lambda)$ to include all pairs $(i,j)$ such that $\sigma_{ij}=\sigma_\lambda $. We set $\Delta(\JJ+1) = ((i,j): \ss_{ij}=0)$ to be all remaining pairs (not included in $\Delta(\lambda)$ for $\lambda \in [\JJ]$).

For $i \in [k]$, the set $\cA_i$ is defined to be all vectors $\ula=(a_1,\dots,a_\ell)$ that satisfy for $j \in [\ell]$:
\begin{itemize}
\item If $\ss_{ij}=1$ and $\sigma_{ij} = \sigma_\lambda$, then $a_j \leq \lambda$; or
\item  If $\ss_{ij}=0$ then $a_j \leq \JJ+1$.
\end{itemize}
In message $m_{(i,\ula)}$, the $j$'th entry $a_j$ of $\ula$ corresponds to the decoding time of $m_{(i,\ula)}$ at receiver $v_j$ in $\cN_0$. 
If $a_j=\lambda$ we require decoding time $\sigma_\lambda$ at $v_j$. To guarantee that a code for $\cN_0$ will also imply one for $\cN$, we require that all parts $m_{(i,\ula)}$ of $m_i$ will be decoded at  $v_j$ at or before their required time $\sigma_{ij}$ in $\cN$.
This latter requirement is met by setting $a_j \leq \lambda$ when $\sigma_{ij}=\sigma_\lambda$ (according to the first bullet above).
Moreover, in $\cN_0$, we may require that part  $m_{(i,\ula)}$ of message $m_i$ be decoded at $v_j$, even when message $m_i$ is not required at all by $v_j$ in $\cN$ (i.e., $\ss_{ij}=0$).  
Namely, for $\ss_{ij}=0$ and a given message $m_{(i,\ula)}$, we may set $a_j$ in $\ula$ to be equal to $\lambda \in [\JJ]$ to represent a decoding time of $\sigma_\lambda$.
The critical point here is, that for $\ss_{ij}=0$, we can choose $m_{(i,\ula)}$ with $a_j = \JJ+1$ to represent that $v_j$ does not require $m_{(i,\ula)}$ or with $a_j \in [\JJ]$ to allow for the possibility that it might be useful for $v_j$ to learn some or all of $m_i$ (for example to be used as side-information).

We now formalize the discussion above.
We define $\cW_0$, $\ulM_0$, $\HH_0$, $\SS_0$ and $\uls_0$. 
\begin{itemize}
\item The network $\cW_0$ is identical to $\cW$.
\item The message set $\ulM_0$ is equal to $(m_{(i,\ula)}: i \in [k], \ula \in \cA_i)$.
\item Define matrix $\SS_0$ with entries $\ss^{(0)}_{(i,\ula),j}$ and vector $\uls_0$ with entries $\sigma^{(0)}_{(i,\ula),j}$ as follows.
For $i \in [k]$, $\ula=(a_1,\dots,a_\ell) \in \cA_i$, and $j \in [\ell]$, if $a_j=\lambda \in [\JJ]$ set 
$\ss^{(0)}_{(i,\ula),j}=1$ and $\sigma^{(0)}_{(i,\ula),j}=\sigma_\lambda$;
otherwise,  set $\ss^{(0)}_{(i,\ula),j}=0$. 
\item Define matrix $\HH_0$ with entries $\hh^{(0)}_{(i,\ula),j}$ as follows.
For $i \in [k]$, $\ula \in \cA_i$, and $j \in [\ell]$, define $\hh^{(0)}_{(i,\ula),j}=1$ if and only if $h_{ij}=1$ and $\hh^{(0)}_{(i,\ula),j}=0$ otherwise. 
\end{itemize}

An example of our definitions on a $3$-node network are given below and depicted in Figure~\ref{fig:M0}.
\begin{example}
\label{ex:N0}
Consider a network $\cN$ with nodes $v_1$, $v_2$, and $v_3$, and a single message $m_1$. 
Suppose that $m_1$ is required by node $v_1$ at time-parameter $\sigma_{1,1}$ and by node $v_2$ at time-parameter $\sigma_{1,2}$; $m_1$ is not required by node $v_3$.
Let $\sigma_{1,1}<\sigma_{1,2}$; then $\JJ=2$ and $\sigma_1=\sigma_{1,1}$, $\sigma_2=\sigma_{1,2}$. 
Message $m_1$ is partitioned into sub-messages $m_{1,\ula}$ for $\ula \in \cA_1 \subseteq [\JJ+1]^\ell=[3]^3$. The set $\cA_1$ includes all $\ula=(a_1,a_2,a_3)$ such that $a_1 \leq 1$ (as $\sigma_1 = \sigma_{1,1}$), $a_2 \leq 2$ (as $\sigma_2 = \sigma_{1,2}$), and $a_3 \leq \JJ+1=3$ (as $m_1$ is not required by node $v_3$ in $\cN$). 
In Figure~\ref{fig:M0}, each sub-message $m_{1,\ula}$ is depicted by a colored box. 
For example, the blue box represents the message $m_{1,(1,1,1)}$. 
For each message $m_{1,\ula}$ the vector $\ula$ specifies the decoding time of  message $m_{1,\ula}$ at nodes $v_1,v_2,v_3$.
That is, if $\ula=(a_1,a_2,a_3)$ then for each $j$ with $a_j \leq \JJ$, $m_{1,\ula}$ is to be decoded by node $v_j$ by time $\sigma_{a_j}n$ while for $a_j=\JJ+1=3$, $m_{1,\ula}$ is not required by node $v_j$. 
For example, message $m_{1,(1,2,3)}$, represented by a white box, must be decoded by time $\sigma_1 n=\sigma_{1,1}n$ at node $v_1$,  by time $\sigma_2 n=\sigma_{1,2}n$ at node $v_2$, and is not required by $v_3$.
Thus a white box is placed before $\sigma_{1,1}n$ on the horizontal {\em time-line} for $v_1$ and before $\sigma_{1,2}n$ on the time-line for $v_2$ (but is absent from the time-line for $v_3$).
Similarly, message $m_{1,(1,2,1)}$, represented by a brown box, is to be decoded by time $\sigma_1 n=\sigma_{1,1}n$ at node $v_1$,  by time $\sigma_2 n=\sigma_{1,2}n$ at node $v_2$, and by time $\sigma_1 n=\sigma_{1,1}n$ at node $v_3$. The locations of the brown  box on the time-lines of nodes $v_1,v_2,v_3$ appear accordingly.
\end{example}

We now have the following theorem:
\begin{theorem}
\label{the:N0}
Let $\cN=(\cW,\ulM,\HH,\SS)$.
Let $\uls = (\sigma_{ij}: \ss_{ij}=1)$ be a collection of time parameters.
Let $\JJ$ be the number of distinct values in $\uls$.
Let $\cN_0=(\cW_0,\ulM_0,\HH_0,\SS_0)$ and $\uls_0$ be defined as above.
For rate vector $\ulR_0 = (R^{(0)}_{(i,\ula)} : i \in [k], \ula \in \cA_i)$, let 
$\ulR = (R_i : i\in[k])$ be defined by $R_i=\sum_{\ula}R^{(0)}_{(i,\ula)}$ for all $i \in [k]$.
Then,
$$
\ulR_0 \in \cRa_{\uls_0}(\cN_0) \ \ \Rightarrow\ \ \ulR \in \cRa_\uls(\cN) .
$$
Moreover, for rate vector $\ulR = (R_i : i\in[k])$, let $\ulR_0 = (R^{(0)}_{(i,\ula)} : i \in [k], \ula \in \cA_i)$ satisfy $R^{(0)}_{(i,\ula)}=R_i$ for all $i \in [k]$ and the unique $\ula=(a_1,\dots,a_\ell)$ in which $a_j = \lambda$  if $\sigma_{ij}=\sigma_\lambda$ and $a_j=\JJ+1$ otherwise, and let $R^{(0)}_{(i,\ula)}=0$ otherwise, then
$$
\ulR \in \cRa_\uls(\cN)  \ \ \Rightarrow\ \ \ulR_0 \in \cRa_{\uls_0}(\cN_0).
$$

\end{theorem}

\begin{proof}
Let $\cR_0 = (R^{(0)}_{(i,\ula)} : i \in [k], \ula \in \cA_i)$ satisfy $R_i=\sum_{\ula}R^{(0)}_{(i,\ula)}$.
To prove that $\ulR_0 \in \cRa_{\uls_0}(\cN_0)$ implies $\ulR \in \cRa_\uls(\cN)$,
consider any $(\e,\uls_0,\ulR_0,n)$ code $\cC_0$ for $\cN_0$. 
Using the exact same code on $\cN$ where each message $m_i$ of $\cN$ is taken to be the concatenation of the messages $(m_{(i,\ula)} : \ula \in \cA_i)$ of $\cN_0$ results in an  $(\e,\uls,\ulR,n)$ code $\cC$ for $\cN$. 
Specifically, by our definition of $\cA_i$ and $\uls_0$, if $\ss_{ij}=1$ then each part $m_{(i,\ula)}$ of message $m_i$ is decoded by $v_j$ at or before time $\sigma_{ij}n$.

For the other direction, consider any $(\e,\uls,\ulR,n)$ code $\cC$ for $\cN$.
Let $\ulR_0$ be as defined in the theorem statement.
For all $R^{(0)}_{(i,\ula)}$ that equal $R_i$, message $m_{(i,\ula)}$ has decoding times identical to those of $m_i$ in $\cN$, and thus can be communicated in $\cN_0$ using $\cC$.
Moreover, all other $R^{(0)}_{(i,\ula)}$ equal $0$.
Thus, the exact same code $\cC$ is an $(\e,\uls_0,\ulR_0,n)$ code for $\cN_0$ as well.
\end{proof}
\vspace{2mm}

\noindent
{\bf $\bullet$ Expanding $\cN_0$:}
Let $\cN=(\cW,\ulM,\HH,\SS)$ and $\cN_0=(\cW_0,\ulM_0,\HH_0,\SS_0)$ be as defined above.
We now present an expansion of $\cN_0$ to network problems $\cN_1,\dots,\cN_\JJ$.
As with the definition of $\cN_0$, the expansion depends on $\uls$.
Let $\JJ$ be the number of distinct values in $\uls$, and let $\sigma_1,\dots,\sigma_\JJ$ and $\Delta(1),\dots,\Delta(\JJ+1)$ be defined as before.
Our goal in defining networks $\cN_1,\dots,\cN_\JJ$ is to better capture the communication process over $\cN_0$.
The networks $\cN_\lambda$ for $\lambda \in [\JJ]$ differ from $\cN_0$ only in the requirements $\SS_\lambda$ and in the side-information $\HH_\lambda$.


Namely,
\begin{itemize}
\item For $\lambda \in [\JJ]$, $\cN_\lambda=(\cW_0,\ulM_0,\HH_\lambda,\SS_\lambda)$.
\item Define matrix $\SS_\lambda$ with entries $\ss^{(\lambda)}_{(i,\ula),j}$ as follows.
For $\lambda \in [\JJ]$, $i \in [k]$, $\ula \in \cA_i$, and $j \in [\ell]$, let $\ss^{(\lambda)}_{(i,\ula),j}=1$ for all pairs $((i,\ula),j)$ with $\ula=(a_1,\dots,a_\ell)$ such that $a_j=\lambda$. 
Otherwise set $\ss^{(\lambda)}_{(i,\ula),j}=0$. 
Notice that $$\sum_{\lambda=1}^\JJ\SS_\lambda = \SS_0.$$
\item Define matrix $\HH_\lambda$ with entries $\hh^{(\lambda)}_{(i,\ula),j}$ as follows.
For $\lambda \in [\JJ]$, $i \in [k]$, $\ula \in \cA_i$, and $j \in [\ell]$,
$$
 \hh^{(\lambda)}_{(i,\ula),j} = \hh^{(0)}_{(i,\ula),j} \cup \sum_{\lambda'=1}^{\lambda-1}\ss^{(\lambda')}_{(i,\ula),j}.
$$
Here, `$\cup$', represents binary logical {\tt OR}.
Namely, in $\cN_\lambda$, node $v_j$ holds all sub-messages $m_{i,\ula}$ that it holds in $\cN_0$ and, in addition, it holds all sub-messages $m_{i,\ula}$ that are required to be decoded at $v_j$ in $\cN_0$ before time $\sigma_\lambda$.
\end{itemize}


\begin{theorem}
\label{the:N_lambda}
Let $\cN=(\cW,\ulM,\HH,\SS)$.
Let $\uls = (\sigma_{ij}: \ss_{ij}=1)$ be a collection of time parameters.
Let $\cN_0=(\cW_0,\ulM_0,\HH_0,\SS_0)$ and $\uls_0$ be defined as above.
Let $\JJ$ be the number of distinct values $\sigma_1< \sigma_2 < \dots<\sigma_\JJ$ in $\uls$.
Let $\sigma_0=0$.
For $\lambda \in [\JJ]$, let $\cN_\lambda=(\cW_0,\ulM_0,\HH_\lambda,\SS_\lambda)$ be defined as above. 
Let $\ulR_0 = (R^{(0)}_{(i,\ula)} : i \in [k], \ula \in \cA_i)$ be a rate vector.
For $\lambda \in [\JJ]$, let $\ulR_\lambda = (R^{(\lambda)}_{(i,\ula)} : i \in [k], \ula \in \cA_i)$ be defined by 
$R^{(\lambda)}_{(i,\ula)}=R^{(0)}_{(i,\ula)}$ if there exists $j \in [\ell]$ such that $\ss^{(\lambda)}_{(i,\ula),j}=1$ and $R^{(\lambda)}_{(i,\ula)}=0$ otherwise.
If for $\lambda \in [\JJ]$ it holds that 
$$
\frac{\ulR_\lambda}{\sigma_\lambda-\sigma_{\lambda-1}} \in  \cTone(\cN_\lambda),
$$
then
$$
 \ulR_0 \in \cRa_{\uls_0}(\cN_0).
$$
\end{theorem}

\begin{proof}
Let $n$ be an integer to be determined later.
Let $n_\lambda = n(\sigma_\lambda-\sigma_{\lambda-1})$.
For $\lambda \in [\JJ]$, consider any collection of $(\e,{\bf \underline{1}},\frac{\ulR_\lambda}{\sigma_\lambda-\sigma_{\lambda-1}} ,n_\lambda)$ codes $\cC_\lambda=(\cE_\lambda,\cD_\lambda)$ for $\cN_\lambda$. 
Here, we take $n$ to be as large as needed to ensure that $n_\lambda$ for each $\lambda \in [\JJ]$ in turn is sufficiently large to allow the existence of codes $\cC_1,\dots,\cC_\JJ$.
By Claim~\ref{claim:trade}, our definitions imply for $\lambda \in [\JJ]$ that $\cC_\lambda$ is also an $(\e,(\sigma_\lambda-\sigma_{\lambda-1}){\bf \underline{1}},\ulR_\lambda,n)$ code for $\cN_\lambda$.

For $\lambda \in [\JJ]$, let $\cE^*_\lambda =(\cE_\lambda^\tau: \tau \in[(\sigma_\lambda-\sigma_{\lambda-1})n])$ represent the first $(\sigma_\lambda-\sigma_{\lambda-1})n$ encoders in $\cC_\lambda$. Similarly for  $\cD^*_\lambda$. The code $\cC_0=(\cE_0,\cD_0)$ for $\cN_0$ is defined to consist of the concatenations $\cE_0=\cE^*_1 \circ \cE^*_2 \circ \dots \circ \cE^*_\JJ$ and $\cD_0=\cD^*_1 \circ \cD^*_2 \circ \dots \circ \cD^*_\JJ$. 

We now show by induction that $\cC_0$ is an $(\JJ\e,\uls_0,\ulR_0,n)$ code for $\cN_0$.
First consider all messages $m_{(i,\ula)}$ and nodes $v_j$ such that for $\lambda=1$, $\ss^{(\lambda)}_{(i,\ula),j}=1$ and thus $\sigma^{(0)}_{(i,\ula),j}=\sigma_\lambda=\sigma_1$.
By our definitions of $\cC_1$ and $\cN_1$, with error probability at most $\e$, for all such pairs $((i,\ula),j))$, message $m_{(i,\ula)}$ is decoded successfully by $v_j$ after $\sigma_{1} n=(\sigma_1-\sigma_0)n$ time steps of $\cC_0$. 
This implies, for $\lambda=1$, that at time $\sigma_\lambda n$,  if $\hh^{(\lambda+1)}_{(i,\ula),j}=1$ then $v_j$ holds  message $m_{(i,\ula)}$.
Recall that,
$$
\hh^{(\lambda+1)}_{(i,\ula),j} = \hh^{(0)}_{(i,\ula),j} \cup \sum_{\lambda'=1}^{\lambda}\ss^{(\lambda')}_{(i,\ula),j}.
$$


%
 We continue by induction.
 Assume, for $\lambda-1$, that after $\sigma_{\lambda-1} n$ time steps of $\cC_0$, with error probability at most $(\lambda-1)\e$ if $\hh^{(\lambda)}_{(i,\ula),j}=1$ then $v_j$ holds  message $m_{(i,\ula)}$.
 We wish to prove the corresponding statement for $\lambda$.
As with the base case of $\lambda=1$, consider all messages $m_{(i,\ula)}$ and nodes $v_j$ such that $\ss^{(\lambda)}_{(i,\ula),j}=1$ and thus $\sigma^{(0)}_{(i,\ula),j}=\sigma_\lambda$.
By induction, after $\sigma_{\lambda-1}n$ time steps of $\cC_0$,  with error probability at most $(\lambda-1)\e$, if $\hh^{(\lambda)}_{(i,\ula),j}=1$ then node $v_j$ holds  message $m_{(i,\ula)}$.
Code $\cC_0$ in time steps $(\sigma_{\lambda}-\sigma_{\lambda-1})n$ employs the encoder and decoder of $\cC_\lambda$, which is an $(\e,(\sigma_\lambda-\sigma_{\lambda-1}){\bf \underline{1}},\ulR_\lambda,n)$ code for $\cN_\lambda$. 
By our definitions, during those time steps of $\cC_0$, with error probability at most $\e$, for all pairs $((i,\ula),j))$ for which $\ss^{(\lambda)}_{(i,\ula),j}=1$ message $m_{(i,\ula)}$ is decoded by $v_j$. 
Thus by the union bound, with overall error probability of at most $\lambda \e$, after the first $\sigma_\lambda n$ time steps of $\cC_0$, if $\hh^{(\lambda+1)}_{(i,\ula),j}=1$ then $v_j$ holds  message $m_{(i,\ula)}$.
%

Continuing this process until $\lambda=\JJ$, and using our definitions of $\SS_0$, $\SS_\lambda$, $\HH_\lambda$, and $\ulR_\lambda$ for $\lambda \in [\JJ]$, we conclude that, employing $\cC_0$, with probability at least $1-\JJ\e$ all nodes $v_j$  decode their required information according to $\uls_0$. Thus, $\cC_0$ is an $(\JJ\e,\uls_0,\ulR_0,n)$ code for $\cN_0$.
\end{proof}
Combining Theorems~\ref{the:N0} and \ref{the:N_lambda} now implies an inner bound to $\cRa_\uls(\cN)$ derived from $(\cTone(\cN_\lambda) : \lambda \in [\JJ])$.
\begin{remark}
In the design of $\cN_0,\cN_1,\dots,\cN_\JJ$, the message set $\ulM_0$ included messages $(m_{i,\ula} : i \in [k], \ula \in \cA_i)$, where for each $i \in [k]$, in the proof of Theorem~\ref{the:N0}, the collection $(m_{i,\ula} : \ula \in \cA_i)$ is considered to be a partition of the original message $m_i \in \ulM$ of $\cN$.
We note that partitioning the message $m_i$ into {\em more} sub-messages does not imply a stronger inner-bound on $\cT(\cN)$. 
Specifically, consider a partition $(m_{i,b} : b \in \cB_i)$ for some index set $\cB_i$ for which $|\cB_i| > |\cA_i|$. 
By the pigeonhole principle, no matter how we set the decoding times for messages $(m_{i,b} : b \in \cB_i)$, there exist at least two messages $m_{i,b_1}$ and $m_{i,b_2}$ with identical corresponding decoding times (at each of the nodes $v_j \in V$). For the purposes of the proofs of Theorems~\ref{the:N0} and \ref{the:N_lambda}, these messages can be merged into a single one. Continuing in this manner, one can formally  show that any rate achievable using  the proofs of Theorems~\ref{the:N0} and \ref{the:N_lambda} with the message set $(m_{i,b} : i \in [k], b \in \cB_i)$ can also be achieved using the original message set $(m_{i,\ula} : i \in [k], \ula \in \cA_i)$.

Moreover, in the design of $\cN_0,\cN_1,\dots,\cN_\JJ$ and in the proofs of Theorems~\ref{the:N0} and \ref{the:N_lambda}, for a given $\uls$, we consider $\JJ$ rounds of communication in which round $\lambda \in [\JJ]$ is done over $\cN_\lambda$. Here, $\JJ$ equals the number of distinct time-parameters in $\uls$.
One could consider increasing the number of rounds of communication beyond the number of distinct time-parameters in $\uls$ (and modifying the message set $\ulM_0$ and the networks $\cN_0,\cN_1,\dots$ accordingly). 
We note, similar to the discussion above, that modifying Theorems~\ref{the:N0} and \ref{the:N_lambda} in this manner also does not yield a better inner-bound on $\cT(\cN)$.

Specifically, using the notation of Theorems~\ref{the:N0} and \ref{the:N_lambda} and their proofs,  consider, for example, partitioning a specific round $\lambda$ into two phases, denoted $\lambda^{(1)}$ and $\lambda^{(2)}$, the first phase taking place during time-steps $n(\sigma_{\lambda^{(1)}}-\sigma_{\lambda-1})$ of $\cC_0$ and the second during time steps  $n(\sigma_\lambda-\sigma_{\lambda^{(1)}})$ for a new parameter $\sigma_{\lambda^{(1)}} \in (\sigma_{\lambda-1},\sigma_\lambda)$. 
Consider also the corresponding modifications needed in the proofs of Theorems~\ref{the:N0} and \ref{the:N_lambda}, including refining the  message set $\ulM_0$, the sets $\cA_i$ for $i \in [k]$, and replacing $\cN_\lambda$ by two network communication problems $\cN_{\lambda^{(1)}}$ and $\cN_{\lambda^{(2)}}$.
By merging messages $m_{i,\ula}$ and $m_{i,\ula'}$ with $\ula=(a_1,\dots,a_\ell)$ and $\ula'=(a'_1,\dots,a'_\ell)$ that differ only in locations $j \in [\ell]$ for which $a_j=\lambda^{(1)}$ and $a'_j = \lambda^{(2)}$ or vice versa,
it can be shown that any achievable rate for $\cN_0$ in the modified proof of Theorem~\ref{the:N_lambda} can also be obtained in the original proof. 
\end{remark}

\section{Theorem~\ref{the:N_lambda} is not tight}
\label{sec:results_example}
It may not come as a surprise to the reader that Theorem~\ref{the:N_lambda} is not tight 
(even for {\em single-message} 
networks $\cN$).
For completeness, we present a single-message network communication problem $\cN$ for which Theorem~\ref{the:N_lambda} is not tight. 
Consider the memoryless degraded 2-terminal broadcast channel characterized by $W(y_1,y_2|x)$ where $y_1=x$ is the binary identity channel, and $y_2=BEC_{0.5}(x)$ corresponds to the binary erasure channel in which $y_2=x$ with probability $0.5$ and $y_2$ is an erasure symbol $\perp$ otherwise.  Namely, $\cW$ is defined by three nodes that are denoted here as the encoder $u$ and two receivers $t_1$ and $t_2$.
The message set $\ulM$ includes a single message $m$, available at node $u$ and required by nodes $t_1$ and $t_2$.
Consider the setting in which $\uls = (\sigma_1,\sigma_2)$, where $\sigma_1=0.5$ is the decoding time required at $t_1$ and $\sigma_2=1$ the decoding time at $t_2$.
As in Section~\ref{sec:static}, for clarity of presentation, we use the common notation for the broadcast channel, instead of that given in Section~\ref{sec:model}.

On one hand, using the characterization of \cite{shulman2000static,Shulman03} presented in Section~\ref{sec:static}, it holds that $\cRa_\uls(\cN)=[0,0.5]$.
The optimal rate $R=0.5$ is obtained, for example, 
using a codebook chosen uniformly at random, i.e., using the uniform distribution $\bx$ over $\{0,1\}$ in the characterization 
\begin{align*}
R &\leq \sigma_{1} I(\bx;\by_1)=\sigma_1\cdot 1=0.5,\\
R &\leq \sigma_{1} I(\bx;\by_2)+(\sigma_{2}-\sigma_{1})C_2=\sigma_1 \cdot 0.5+0.25=0.5
\end{align*}

On the other hand, applying Theorems~\ref{the:N0} and \ref{the:N_lambda}, the message set $\ulM_0$ of $\cN_0$ includes two messages $m_{(1,1)}$ and $m_{(1,2)}$ with corresponding decoding times $\uls_0$ of $\sigma^{(0)}_{(1,1),1}=\sigma^{(0)}_{(1,1),2}=\sigma^{(0)}_{(1,2),1}=\sigma_1=0.5$ and $\sigma^{(0)}_{(1,2),2}=\sigma_2=1$. 
Consider any $\ulR_0 = (R^{(0)}_{(1,1)},R^{(0)}_{(1,2)}) \in \cRa_{\uls_0}(\cN_0)$.
By Theorem~\ref{the:N0}, it holds that the sum-rate $R^{(0)}_{(1,1)} + R^{(0)}_{(1,2)}$ can be as large as $R=0.5$.
We now show that for any $\ulR_0^*= (R^{(*)}_{(1,1)},R^{(*)}_{(1,2)})$ satisfying the conditions of Theorem~\ref{the:N_lambda} it holds that $R^{(*)}_{(1,1)} + R^{(*)}_{(1,2)}<0.5$, implying that Theorem~\ref{the:N_lambda} is not tight.

Consider the network communication problem $\cN_1$ corresponding to the first phase of communication in $\cN_0$.
Namely, $\cN_1$ takes into consideration the requirements $\ss^{(1)}_{(1,1),1}=\ss^{(1)}_{(1,1),2}=\ss^{(1)}_{(1,2),1}=1$, and $\ss^{(1)}_{(1,2),2}=0$.
It can be seen that 
$2R^{(1)}_{(1,1)}+R^{(1)}_{(1,2)} \leq 0.5$, for any $\frac{\ulR_1}{\sigma_1}=2\ulR_1=(2R^{(1)}_{(1,1)},2R^{(1)}_{(1,2)}) \in \cTone(\cN_1)$.
This follows from studying the capacity of the (degraded) broadcast channel $W$ with common message $m_{(1,1)}$ and private message $m_{(1,2)}$ at terminal $t_1$.
By our definitions in Theorem~\ref{the:N_lambda} it holds that $(R^{(1)}_{(1,1)},R^{(1)}_{(1,2)})=(R^{(*)}_{(1,1)},R^{(*)}_{(1,2)})$.
Thus, either $R^{(*)}_{(1,1)}=0$, or alternatively, $R^{(*)}_{(1,1)}+R^{(*)}_{(1,2)}<0.5$.
In the latter, we conclude that Theorem~\ref{the:N_lambda} is not tight.
For the former, we further study $R^{(*)}_{(1,2)}$. This time we consider $\cN_2$ and see that $R^{(2)}_{(1,2)}=R^{(*)}_{(1,2)} \in (\sigma_2-\sigma_1)\cTone(\cN_2)=\frac{1}{2}\cTone(\cN_2)$ implies $R^{(*)}_{(1,2)} \leq 0.25$. Which in turn implies that $R^{(*)}_{(1,1)}+R^{(*)}_{(1,2)} \leq 0+0.25<0.5$ and thus that Theorem~\ref{the:N_lambda} is not tight in this case as well.

\section{Conclusions}
\label{sec:conclude}
In this work we generalize the standard notion of capacity by studying the time-rate region $\cT$ of discrete memoryless networks $\cN$.
We present an inner bound on $\cT$ based on the concatenation of a series of block-codes corresponding to a time-expansion of $\cN$.
Improving on these bounds, for general networks $\cN$, or for specific network components or network communication settings, is the subject to future work.


\bibliographystyle{unsrt}
\bibliography{proposal,online_rateless}
\end{document}